\documentclass{article}
\usepackage{graphicx} 
\usepackage{amssymb}
\usepackage{url}
\usepackage{amsthm, amsmath}
\usepackage{subfig}
\usepackage{color}
\usepackage{lineno}
\usepackage{oplotsymbl}

\newtheorem{theorem}{Theorem}
\newtheorem{corollary}[theorem]{Corollary}

\newtheorem{lemma}[theorem]{Lemma}
\newtheorem{property}[theorem]{Property}

\newtheorem*{remark}{Remark}

\title{Sibson's formula for higher order Voronoi diagrams}
\author{Merc\`e Claverol, Andrea de las Heras-Parrilla,\\Clemens Huemer and Dolores Lara}
\date{}

\begin{document}

\maketitle

\begin{abstract}
    Let $S$ be a set of $n$ points in general position in $\mathbb{R}^d$. The order-$k$ Voronoi diagram of $S$, $V_k(S)$, is a subdivision of $\mathbb{R}^d$ into cells whose points have the same $k$ nearest points of $S$.
    Sibson, in his seminal paper from 1980 (A vector identity for the Dirichlet tessellation), gives a formula to express a point $Q$ of $S$ as a convex combination of other points of $S$ by using ratios of volumes of the intersection of cells of $V_2(S)$ and the cell of $Q$ in $V_1(S)$. The natural neighbour interpolation method is based on Sibson's formula. We generalize his result to express $Q$ as a convex combination of other points of $S$ by using ratios of volumes from Voronoi diagrams of any given order.
\end{abstract}
\section{Introduction}

Let $S$ be a set of $n$ points in general position in $\mathbb{R}^d$, meaning no $m$ of them lie in a $(m-2)$-dimensional flat for $m=2,3,...,d+1$ and no $d+2$ of them lie in the same $d$-sphere, and let $k$ be a natural number with $1\leq k \leq n-1$.
Let $\sigma_d$ denote the Lebesgue measure on $\mathbb{R}^d$, to simplify we just write $\sigma$.

The order-$k$ Voronoi diagram of $S$, $V_k(S)$, is a subdivision of $\mathbb{R}^d$ into cells such that points in the same cell have the same $k$ nearest points of $S$.
Thus, each cell $f(P_k)$ of $V_k(S)$ is defined by a subset $P_k$ of $S$ of $k$ elements, where each point of $f(P_k)$ has $P_k$ as its $k$ closest points from $S$. See Figure~\ref{fig:F4-V}.

\begin{figure}[!ht]
	\centering
	\includegraphics[scale=0.6]{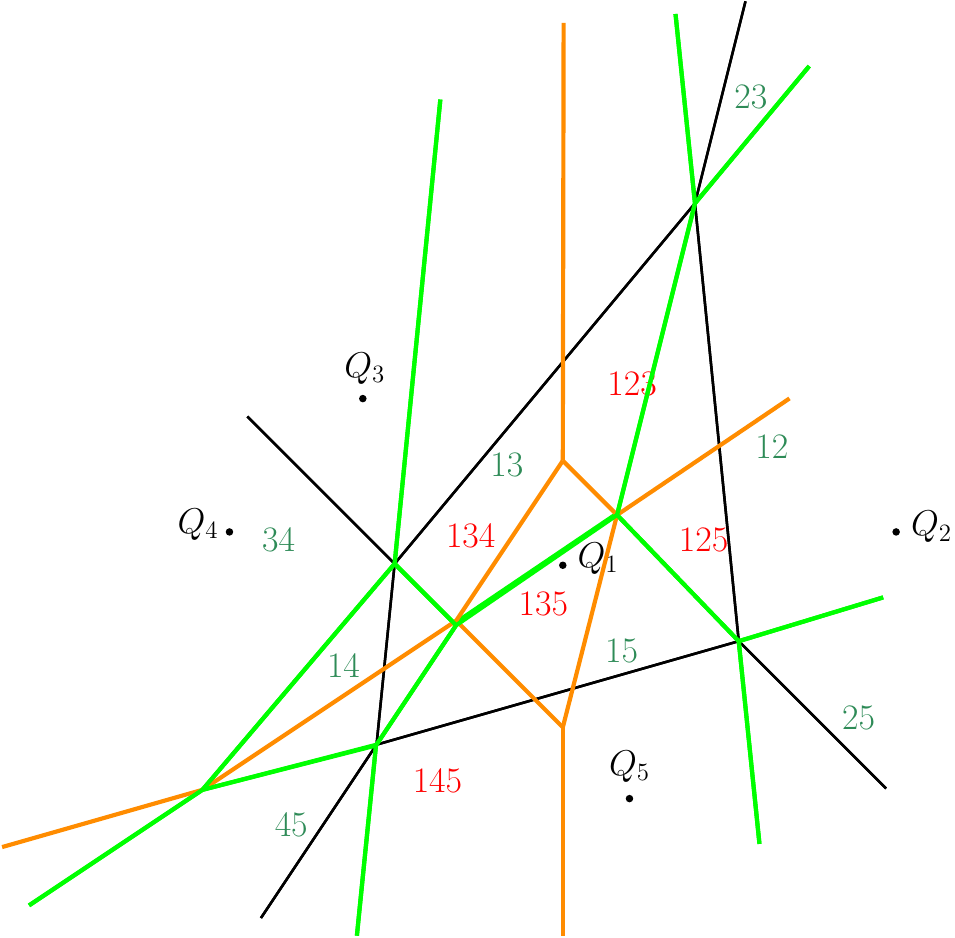}
	\caption{
        For a set $S=\{Q_1,\cdots,Q_5\}$ of five points in $\mathbb{R}^2$.
        $V_1(S)$ is shown in black, $V_2(S)$ in green, and $V_3(S)$ in orange colour.
		Each cell of $V_2(S)$ ($V_3(S)$) is labeled by the indices of its two (three) nearest points of $S$. 
	}
	\label{fig:F4-V}
\end{figure}

For the order-$k$ Voronoi diagram of $S$, the region $R_k(\ell)$ of $Q_\ell\in S$ is defined as the set of cells of $V_k(S)$ that have the point $Q_\ell$ as one of their $k$ nearest neighbours. 
See Figure~\ref{fig:R1InR2}.
These regions are not necessarily convex but star-shaped, see \cite{aurenhammer1991simple, edelsbrunner2018multiple, okabe2009spatial, toth1976multiple}, and it is known that $R_1(\ell)$ is contained in the kernel of $R_k(\ell)$; see \cite{CdH21}.
Also, these regions are related to Brillouin zones.
For a given $k$, the region $R_k(\ell)\setminus R_{k-1}(\ell)$ is known as a Brillouin zone of $Q_\ell$. Brillouin zones have been studied mainly for lattices but also for arbitrary discrete sets, see e.g. \cite{jones1984geometric, veerman2000brillouin}.

\begin{figure}[!ht]
	\centering	
	\includegraphics[scale=0.5,page=3, trim=1cm 12cm 0.5cm 5cm, clip]{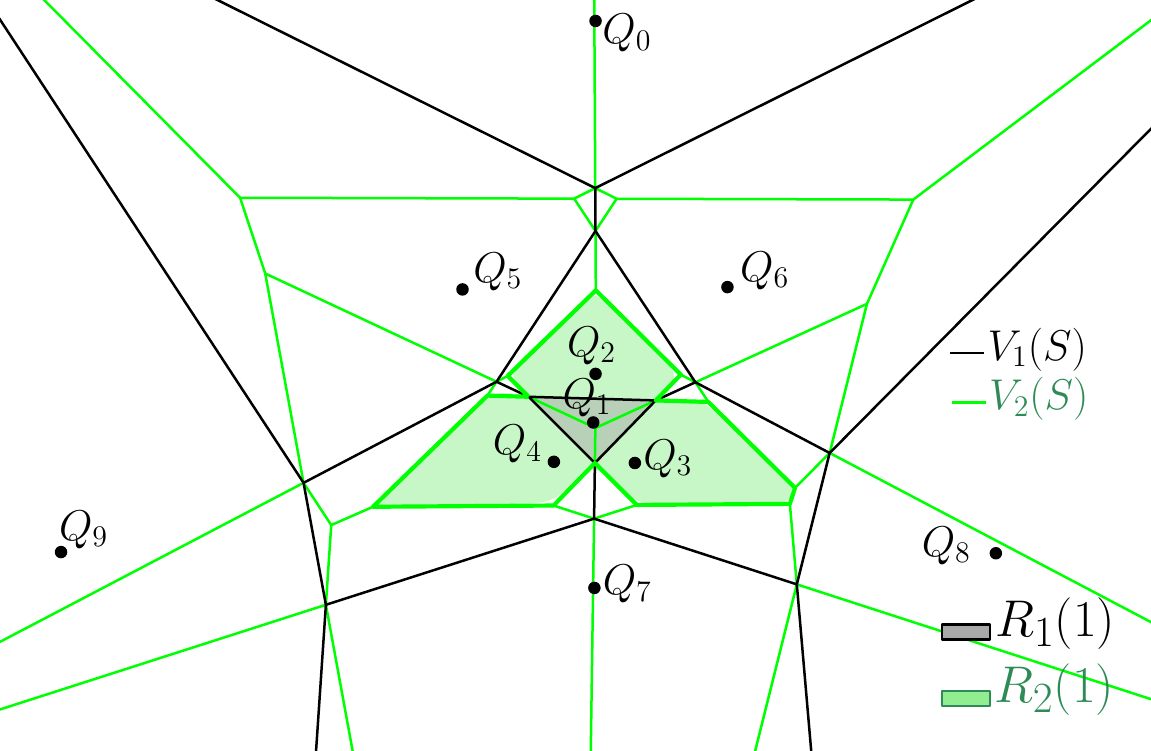}
	\caption{$R_1(1)$ is the cell $f(\{Q_1\})$ in $V_1(S)$. $R_2(1)$ is the union of cells  of $V_2(S)$ that have $Q_1$ as one of its two nearest neighbours. $R_1(1)\subset R_2(1)$.}
    \label{fig:R1InR2}
\end{figure}

Local coordinates based on Voronoi diagrams were introduced by Sibson~\cite{sibson1980vector}. He states that, given a set $S$ of $n$ points of $\mathbb{R}^d$ in general position, a point $Q_\ell\in S$ can be expressed as a convex combination of its nearest points of $S$. This is described next.
Cells of $V_2(S)$ that intersect $f(\{Q_\ell\})$ in $V_1(S)$ are of the form $f(\{Q_\ell,Q_j\})$, i.e., cells defined by $Q_\ell$ and another point $Q_j$, that we call its natural neighbour. 
These intersections give ratios of volumes which are the coefficients multiplying the corresponding natural neighbours in the convex combination that expresses $Q_\ell$. 
Volumes $\sigma(f(\{Q_\ell,Q_j\})\cap f(\{Q_\ell\}))$ are equal to the volumes given by the intersection of the cells of $V_1(S\setminus\{Q_\ell\})$ and $f(Q_\ell)$ in $V_1(S)$, see Figure~\ref{fig:F5-NN}.

\begin{theorem}\label{thm:sibson}
    (Local coordinates property \cite{sibson1980vector}). For a bounded cell $f(\{Q_\ell\})$ of $V_1(S)$, 
    \begin{equation}
        Q_\ell = \sum_{j\neq \ell}\frac{\sigma(f(\{Q_\ell,Q_j\})\cap f(\{Q_\ell\}))}{\sigma(f(\{Q_\ell\}))}Q_j
    \end{equation}
\end{theorem}

Sibson's formula has been used to define the natural neighbour interpolation method~\cite{sibson1981brief}.
Given a set of points and a function, this interpolation method provides a smooth approximation of new points to the function.
Sibson's algorithm uses the closest subset of the input set $S\setminus\{Q_\ell\}$ to interpolate the function value of a query point, $Q_\ell$, and applies weights based on the ratios of volumes provided by Theorem~\ref{thm:sibson}.
Local coordinates and the natural neighbour interpolation method have been studied e.g. in~\cite{farin1990surfaces, piper1993properties,sugihara1999surface}, and they have many applications such as reconstruction of a surface from unstructured data or interpolation of  rainfall data, see \cite{okabe1994nearest, sugihara1999surface}.

\begin{figure}[!ht]
	\centering
	\subfloat[]{
		\includegraphics[scale=0.28,page=1, trim=0cm 0cm 1cm 0cm, clip]{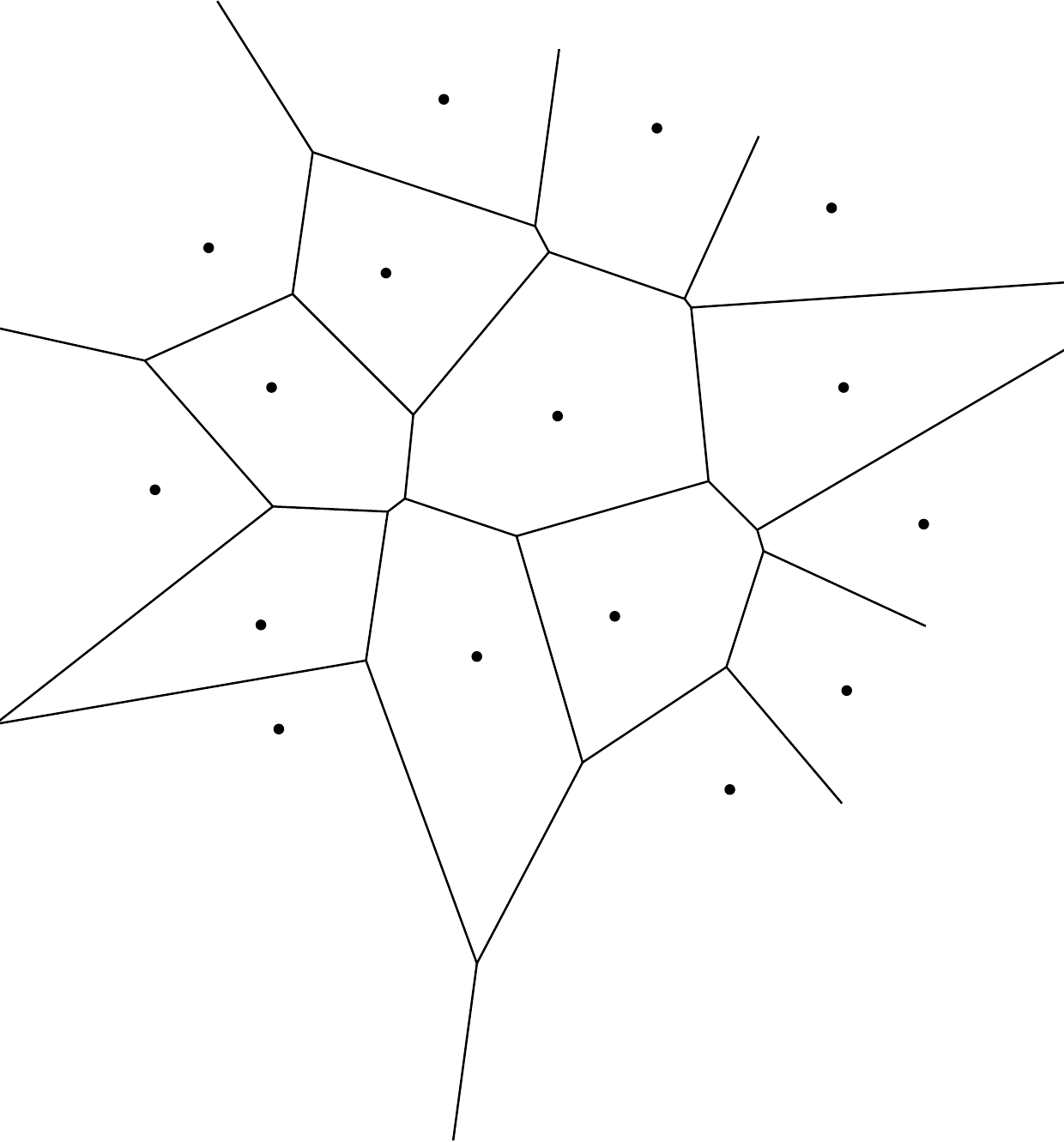}
	}~
	\subfloat[]{
		\includegraphics[scale=0.28,page=2, trim=1cm 2.8cm 0cm 3.5cm, clip]{F5-NaturalNeighbour.pdf}
	}
	\caption{
		In $\mathbb{R}^2$. 
        (a) The initial Voronoi diagram $V_1(S\setminus\{Q_\ell\})$ without query point $Q_\ell$.
		(b) Coloured areas given by the intersections of $f(\{Q_\ell\})$ and the cells of $V_1(S\setminus\{Q_\ell\})$, are the same as the ones given by the intersections of the cells of $V_2(S)$ (shown in dashed) with the cell $f(\{Q_\ell\})$.
	}
	\label{fig:F5-NN}
\end{figure}

Aurenhammer gave a generalization of Sibson's result to Voronoi diagrams of higher order, and more generally to power diagrams, see \cite{aurenhammer1988linear}.
Aurenhammer's formula allows to write a point $Q_\ell$ of $S$ as a linear combination of other points of $S$.
We state this in Theorem~\ref{thm:Sibsonkface} below. The formula in Theorem~\ref{thm:Sibsonkface} is defined in terms of intersections of cells of $V_{k-1}(S)$ and $V_{k+1}(S)$ with a cell of $V_{k}(S)$. This formula works for a bounded cell of $V_k(S)$.

Our main contribution is another generalization of Sibson's result, stated in  Theorem~\ref{thm:SibsonR_k}.
In this theorem, we express a point $Q_\ell\in S$ as a convex combination of its neighbours of $S$ using ratios of volumes in the region $R_k(\ell)$.
Similar to Sibson's formula that required the cell of the point $Q_\ell$ to be bounded, our formula requires its region $R_k(\ell)$ to be bounded.
For the case $k=1$, Theorem~\ref{thm:SibsonR_k} coincides with Theorem~\ref{thm:sibson}.

This paper is organized as follows. 
In Section~\ref{sec:AurenhammerFormula} we revisit the formula of Aurenhammer for higher order Voronoi diagrams and we give a geometric interpretation of the formula.  
Our main result is given in Section~\ref{sec:coordinatesVoronoiDiagrams}, where we detail our generalization of Sibson's formula.
Finally, Section~\ref{sec:interpolatión} is on how our generalization of Sibson's formula from Section~\ref{sec:coordinatesVoronoiDiagrams} could be used for interpolation.

\section{ A revisit of Aurenhammer's formula}\label{sec:AurenhammerFormula}

Next, we state the theorem of Aurenhammer \cite{aurenhammer1988linear} in terms of $V_k(S)$.

\begin{theorem} (\cite{aurenhammer1988linear}) \label{thm:Sibsonkface}
    Let $2\leq k\leq n-2$ and let $f(P_k)$ be a bounded cell of $V_k(S)$. 
    Then,
    \begin{equation*}
        \sum_{\substack{f(P_{k-1}) \in V_{k-1}(S)\\ Q_i \in P_k\setminus P_{k-1}}}\sigma (f(P_{k-1})\cap f(P_k))Q_i=\sum_{\substack{f(P_{k+1}) \in V_{k+1}(S)\\ Q_j \in P_{k+1}\setminus P_k}}\sigma (f(P_{k+1})\cap f(P_k))Q_j
    \end{equation*}
\end{theorem}

\begin{figure}[!ht]
	\centering
	\subfloat[]{
		\includegraphics[scale=0.3, page=1, trim=0cm 15cm 0cm 0cm, clip]{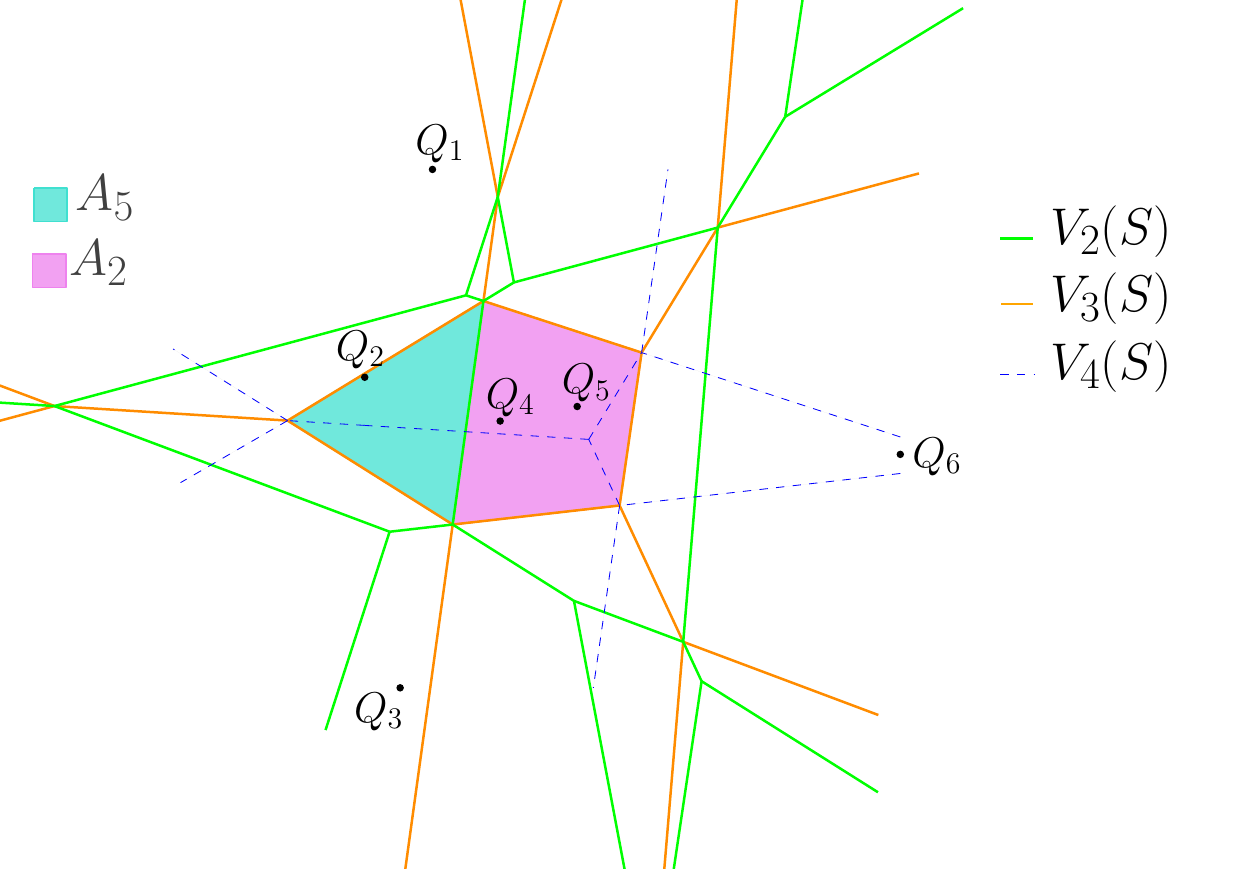}
	}~
	\subfloat[]{
		\includegraphics[scale=0.3, page=2, trim=0cm 15cm 0cm 0cm, clip]{Thm4.pdf}
	}
	\caption{
        Illustrating Theorem~\ref{thm:Sibsonkface} for $f(\{Q_2,Q_4,Q_5\})$ in $V_{3}(S)$, where $S$ is a set of six points in $\mathbb{R}^2$. In this case the equation reduces to $\sigma(A_5) Q_5 + \sigma(A_2)Q_2 =
        \sigma(B_1)Q_1 + \sigma(B_3)Q_3 + \sigma(B_6)Q_6$.
         (a) Regions $A_i$ are the cells of $V_2(S) \cap f(\{Q_2,Q_4,Q_5\})$, whose points have $Q_i$ as the third nearest neighbour of $S$.  
         (b) Regions $B_i$ are the cells of $V_4(S) \cap f(\{Q_2,Q_4,Q_5\})$, whose points have $Q_i$ as the fourth nearest neighbour of $S$.
	}
	\label{fig:thm4Example}
\end{figure}

For an example illustrating Theorem~\ref{thm:Sibsonkface}, see Figure~\ref{fig:thm4Example}.

In the following we examine the generalization of Sibson's theorem to higher order Voronoi diagrams from Theorem~\ref{thm:Sibsonkface} in more detail for  cells $f(P_k)$ of $V_k(S)$, when $S$ is a point set in $\mathbb{R}^2$.
Divide both sides of the equation given in Theorem~\ref{thm:Sibsonkface} by $\sigma(f(P_k))$; then, each side of the equation describes a point $H$ that is a convex combination of points from $S$. 
We have 
\begin{equation}\label{eqn:H}
       H= \hspace{-0.5cm} \sum_{\substack{f(P_{k-1}) \in V_{k-1}(S)\\ Q_i \in P_k\setminus P_{k-1}}} \frac{ \sigma (f(P_{k-1})\cap f(P_k))}{\sigma(f(P_k))}Q_i = \hspace{-0.5cm}\sum_{\substack{f(P_{k+1}) \in V_{k+1}(S)\\ Q_j \in P_{k+1}\setminus P_k}} \frac{\sigma (f(P_{k+1})\cap f(P_k))}{\sigma(f(P_k))}Q_j
    \end{equation}

What can we say about this point $H$?\vspace{0.2cm}

Let $f(P_k)$ be an $r$-gon. Then $S$ contains $r$ points $Q_1, \ldots, Q_r$, such that each edge of the $r$-gon lies on a perpendicular bisector between two of these $r$ points, and each vertex, $C_{ij\ell}$, of $f(P_k)$ is the center of a circle passing through three of them, $Q_i$, $Q_j$, and $Q_\ell$; see e.g.~\cite{CdH21, lee1982Voronoi}.

We denote with $\Delta(ABC)$ the triangle with vertices $A$, $B$, and $C$, and with $\square(ABCD)$ the quadrilateral with vertices $A,B,C$ and $D$, in cyclic order. 

Let us consider the case when $f(P_k)$ is a quadrilateral cell of $V_k(S)$ with vertices $C_{123}$, $C_{124}, C_{134}$, and $C_{234}$, in cyclic order along the boundary of the quadrilateral cell $f(P_k)=\square(C_{123}C_{124}C_{134}C_{234})$. 

One of the diagonals $C_{123} C_{134}$ is an edge of $V_{k-1}(S)$ and the other one, $C_{124}C_{234}$, of $V_{k+1}(S)$.
Figure~\ref{fig:Sibson4gon} shows an example.
We refer to~\cite{CdH21,lee1982Voronoi} for a more detailed discussion on the structure of cells of $V_k(S)$.
Theorem~\ref{thm:Sibsonkface} states in this case that 
\begin{equation}\label{Sibson4gonface}
    \begin{split}
        H&=Q1 \cdot \frac{\sigma(\Delta(C_{123}C_{134}C_{234}))}{\sigma(\square(C_{123}C_{124}C_{134}C_{234}))} + Q3 \cdot \frac{\sigma(\Delta(C_{123}C_{124}C_{134}))}{\sigma(\square(C_{123}C_{124}C_{134}C_{234}))}\\
        & = Q2 \cdot \frac{\sigma(\Delta(C_{124}C_{134}C_{234}))}{\sigma(\square(C_{123}C_{124}C_{134}C_{234}))} + Q4 \cdot \frac{\sigma(\Delta(C_{124}C_{234}C_{123}))}{\sigma(\square(C_{123}C_{124}C_{134}C_{234}))}
    \end{split}
\end{equation}

\begin{figure}[!ht]
    \centering
    \includegraphics[width=0.5\columnwidth]{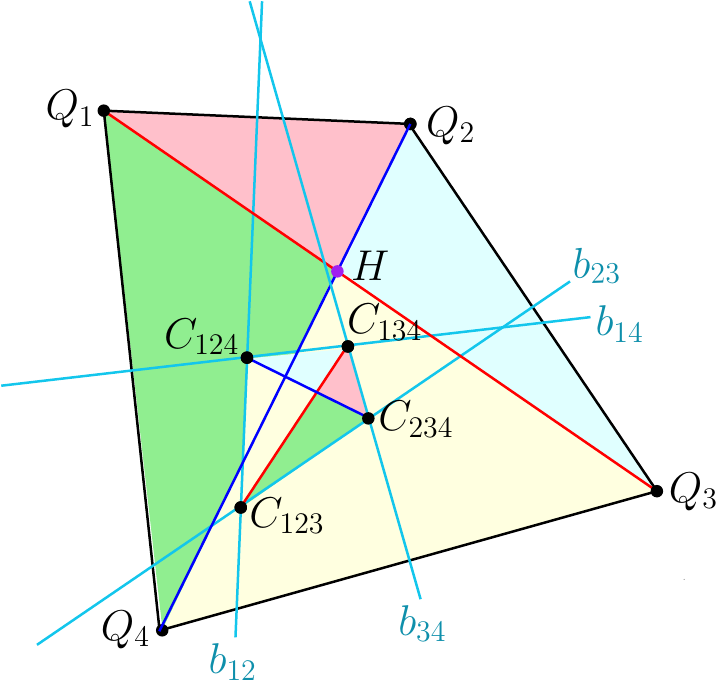}
    \caption{The quadrilateral cell $f(P_k)=\square(C_{123}C_{124}C_{134}C_{234})$ of $V_k(S)$ is obtained by perpendicular bisector construction from $\{Q_1, Q_2, Q_3, Q_4\} \subset S.$ 
    Point $H$ given by Equation~(\ref{eqn:H}) is the intersection point of diagonals $Q_1Q_3$ and $Q_2Q_4.$
    Triangles with same colour have proportional area.
    }
    \label{fig:Sibson4gon}
\end{figure} 

Note that the point $H$ is a convex combination of $Q_1$ and of $Q_3$, and $H$ is also a convex combination of $Q_2$ and $Q_4$, by the right side of Equation~(\ref{Sibson4gonface}).  Then point $H$ is the intersection point of diagonals $Q_1Q_3$ and $Q_2Q_4$ of $\square(Q_1Q_2Q_3Q_4).$

This implies the following corollary.
\begin{corollary}
 Given a quadrilateral cell $\square(C_{123}C_{124}C_{134}C_{234})$ of $V_k(S)$, the four corresponding points $Q_1, Q_2, Q_3, Q_4$ of $S$ that participate in the perpendicular bisectors that define $\square(C_{123}C_{124}C_{134}C_{234})$,  also form a convex quadrilateral, $\square(Q_1Q_2Q_3Q_4).$   
\end{corollary}

More can be said about the areas of triangles with vertices from the set\newline $\{C_{123}, C_{124}, C_{134}, C_{234} \}$, within the  quadrilateral cell $\square(C_{123}C_{124}C_{134}C_{234})$ of $V_k(S).$
We show next, how these triangle areas are related to triangle areas within the  quadrilateral $\square(Q_1Q_2Q_3Q_4).$
As we will see, in the case of a quadrilateral cell of $V_k(S)$,  Theorem~\ref{thm:Sibsonkface}  has an analogous statement for the corresponding  quadrilateral formed by four points of $S$.
Thereto, we recall a folklore result:
\begin{property}\label{prop:triangle}
    Let $P$ be a point contained in a triangle $\Delta(ABC)$. Then $P$ can be expressed as
    \begin{equation*}
        P= A \cdot \frac{\sigma(\Delta(PBC))}{\sigma(\Delta(ABC))} + B \cdot \frac{\sigma(\Delta(PAC))}{\sigma(\Delta(ABC))} + C \cdot \frac{\sigma(\Delta(PAB))}{\sigma(\Delta(ABC))}.
    \end{equation*}
\end{property}

 We also expect next Property~\ref{prop:quadri} to be known. The following proof allows us to infer the relation among triangle areas we want to show.
 
\begin{property}\label{prop:quadri}
    Let $\square(Q_1Q_2Q_3Q_4)$ be a convex quadrilateral. Then, 
    \begin{equation}\label{eqn:quadrilateralnice}
        Q_1 \cdot \sigma(\Delta(Q_2Q_3Q_4)) + Q_3 \cdot \sigma(\Delta(Q_1Q_2Q_4))
        = Q_2 \cdot \sigma(\Delta(Q_1Q_3Q_4)) + Q_4 \cdot \sigma(\Delta(Q_1Q_2Q_3)).
    \end{equation}
\end{property}
\begin{proof}
    Let $H$ be the intersection point of the two diagonals $Q_1Q_3$ and $Q_2Q_4$.
    Apply Property~\ref{prop:triangle} to the triangle $\Delta(Q_1Q_2Q_3)$ and $P=H.$ Then $P$ lies on the edge $Q_1Q_3$, and the degenerate triangle $\Delta(Q_1HQ_3)$ has zero area. It follows that
    \begin{equation}\label{eqn:H1}
        H=Q_1 \cdot \frac{\sigma(\Delta(HQ_2Q_3))}{\sigma(\Delta(Q_1Q_2Q_3))} + Q_3 \cdot \frac{\sigma(\Delta(HQ_1Q_2))}{\sigma(\Delta(Q_1Q_2Q_3))}.
    \end{equation}
    Repeat the same argument for the triangles $\Delta(Q_1Q_3Q_4)$, $\Delta(Q_2Q_3Q_4)$, and $\Delta(Q_1Q_2Q_4)$ to obtain
    \begin{equation}\label{eqn:H2}
        H=Q_1 \cdot \frac{\sigma(\Delta(HQ_3Q_4))}{\sigma(\Delta(Q_1Q_3Q_4))} + Q_3 \cdot \frac{\sigma(\Delta(HQ_1Q_4))}{\sigma(\Delta(Q_1Q_3Q_4))},
    \end{equation}
    \begin{equation}\label{eqn:H3}
        H=Q_2 \cdot \frac{\sigma(\Delta(HQ_3Q_4))}{\sigma(\Delta(Q_2Q_3Q_4))} + Q_4 \cdot \frac{\sigma(\Delta(HQ_2Q_3))}{\sigma(\Delta(Q_2Q_3Q_4))},
    \end{equation}
    \begin{equation}\label{eqn:H4}
        H=Q_2 \cdot \frac{\sigma(\Delta(HQ_1Q_4))}{\sigma(\Delta(Q_1Q_2Q_4))} + Q_4 \cdot \frac{\sigma(\Delta(HQ_1Q_2))}{\sigma(\Delta(Q_1Q_2Q_4))}.
    \end{equation}
    Combine Equations~(\ref{eqn:H1}) and~(\ref{eqn:H2}) to obtain
    
    \begin{equation}\label{eqn:area1}
        H= Q_1 \cdot \frac{\sigma(\Delta(Q_2Q_3Q_4))}{\sigma(\square(Q_1Q_2Q_3Q_4))} + Q_3 \cdot \frac{\sigma(\Delta(Q_1Q_2Q_4))}{\sigma(\square(Q_1Q_2Q_3Q_4))}.
    \end{equation}
    In the same way, combine Equations~(\ref{eqn:H3}) and~(\ref{eqn:H4}) to obtain
    \begin{equation}\label{eqn:area2}
        H= Q_2 \cdot \frac{\sigma(\Delta(Q_1Q_3Q_4))}{\sigma(\square(Q_1Q_2Q_3Q_4))} + Q_4 \cdot \frac{\sigma(\Delta(Q_1Q_2Q_3))}{\sigma(\square(Q_1Q_2Q_3Q_4))}.
    \end{equation}
    Finally, Equation~(\ref{eqn:quadrilateralnice}) follows from Equations~(\ref{eqn:area1}) and~(\ref{eqn:area2}). 

\end{proof}

The quadrilateral cell $\square(C_{123}C_{124}C_{134}C_{234})$ of $V_k(S)$ corresponding to $\square(Q_1Q_2Q_3Q_4)$ can be obtained from the so-called perpendicular bisector construction, see~\cite{radko20128bisector}.
Furthermore, 
\begin{equation}
 \sigma( \square(C_{123}C_{124}C_{134}C_{234})) =|r| \cdot \sigma(\square(Q_1Q_2Q_3Q_4)),
\end{equation}
where  
\begin{equation*}
    r= \frac{1}{4}(\cot(\alpha)+\cot(\gamma))\cdot (\cot(\beta)+\cot(\delta))
\end{equation*}
and $\alpha, \beta, \gamma$ and $\delta$ are the four interior angles of  $\square(Q_1Q_2Q_3Q_4)$ in consecutive order, see~\cite{radko20128bisector}.
From Equations~(\ref{Sibson4gonface}) and~(\ref{eqn:area1}) we see that the coefficient multiplying point $Q_1$ must be the same, then 
\begin{equation*}
    \frac{\sigma(\Delta(C_{123}C_{134}C_{234}))}{\sigma(\square(C_{123}C_{124}C_{134}C_{234}))} =\frac{\sigma(\Delta(Q_2Q_3Q_4))}{\sigma(\square(Q_1Q_2Q_3Q_4))}
\end{equation*}
and
\begin{equation*}
    \sigma(\Delta(C_{123}C_{134}C_{234}))=|r|\cdot \sigma(\Delta(Q_2Q_3Q_4)).
\end{equation*}
The other triangle areas can be related analogously. We also refer to~\cite{mammana2008affine} where it is proved that $\square(Q_1Q_2Q_3Q_4)$ and $\square(C_{123}C_{124}C_{134}C_{234}))$ are affine.\\

Let us then consider the case when $f(P_k)$ is a cell of $V_k(S)$ with more than four sides.
Equation~(\ref{eqn:H}) gives a point $H$ that can be expressed in two ways as convex combination of points of $S$.  Let us look at a pentagonal cell $f(P_k)=\pentago(C_{123}C_{134}C_{145}C_{245}C_{125})$ of $V_k(S)$; See Figure~\ref{fig:F6-pentagon}.   For $r>5$ the situation is similar.  Theorem~\ref{thm:Sibsonkface} here gives 
\begin{equation*}
    \begin{split}        
        H&=Q_1 \cdot \frac{\sigma(\square(C_{123}C_{125}C_{145}C_{134}))}{\sigma(\pentago(C_{123}C_{134}C_{145}C_{245}C_{125}))} + Q_5 \cdot \frac{\sigma(\Delta(C_{125}C_{245}C_{145}))}{\sigma(\pentago(C_{123}C_{134}C_{145}C_{245}C_{125}))}\\
        &= Q_2 \cdot \frac{\sigma(\square(C_{245}C_{125}C_{123}C_{234}))}{\sigma(\pentago(C_{123}C_{134}C_{145}C_{245}C_{125}))} + Q_4 \cdot \frac{\sigma(\Delta(C_{245}C_{234}C_{1345}C_{145}))}{\sigma(\pentago(C_{123}C_{134}C_{145}C_{245}C_{125}))}\\
        &+ Q_3 \cdot \frac{\sigma(\Delta(C_{123}C_{234}C_{134}))}
        {\sigma(\pentago(C_{123}C_{134}C_{145}C_{245}C_{125}))}
    \end{split}
\end{equation*}

We get that $H$ lies on the segment $Q_1Q_5$ and inside the triangle $\Delta(Q_2Q_3Q_4)$. Furthermore, $H$ divides the segment $Q_1Q_5$ in the same proportion as the edge $C_{125}C_{145}$ divides the pentagon $\pentago(C_{123}C_{134}C_{145}C_{245}C_{125})$ into the quadrilateral $\square(C_{125}C_{145}C_{134}C_{123})$ and the triangle $\Delta(C_{125}C_{145}C_{245})$.
And $H$ divides triangle $\Delta(Q_1Q_2Q_3)$ in the same proportion into triangles $\Delta(Q_3HQ_4)$, $\Delta(Q_2HQ_3)$, and $\Delta(Q_2HQ_4)$ as $C_{234}$ divides $\pentago(C_{123}C_{134}C_{145}C_{245}C_{125})$ into $\square(C_{245}C_{125}C_{123}C_{234})$, $\square(C_{245}C_{234}C_{134}C_{145})$ and 
 $\Delta(C_{134}C_{234}C_{123})$.\\

\begin{figure}[!ht]
	\centering
    \subfloat[]{\includegraphics[scale=0.3,page=1]{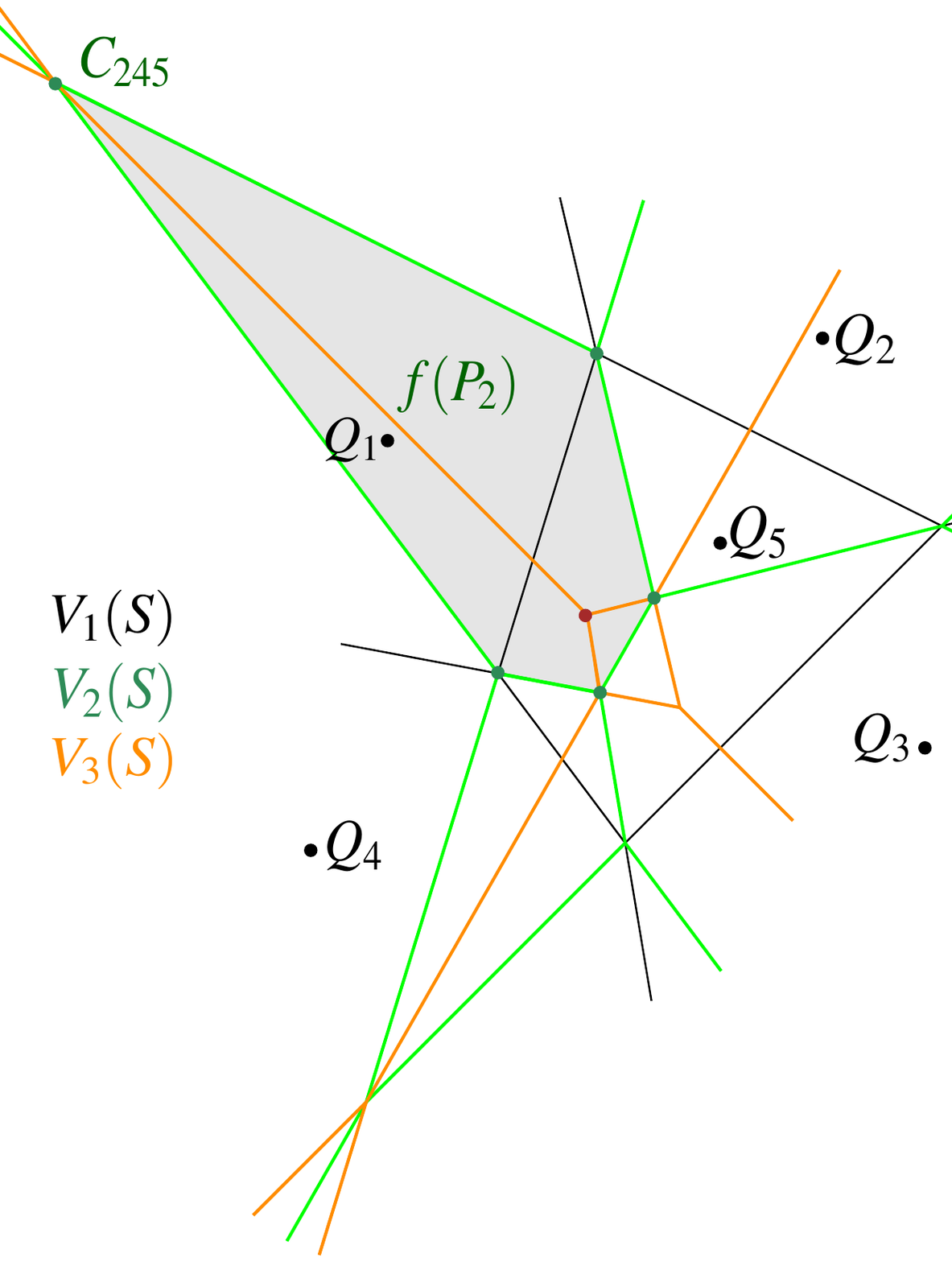}}
    \hspace{-1cm}	
    \subfloat[]{\includegraphics[scale=0.3,page=2]{F6-penta.pdf}}
	\caption{ 
        (a) $V_3(S)$ for a set of five points $S=\{Q_1,Q_2,Q_3, Q_4, Q_5\}.$ For $P_2=\{ Q_1, Q_5\}$, the grey region is the pentagonal cell $f(P_2)$ of $V_2(S).$ $f(P_2)$ is divided by an edge of $V_1(S)$ and is also divided by three edges of $V_3(S)$. (b) The point $H$ lies on the segment $Q_1Q_5$ and inside the triangle $\Delta(Q_2,Q_3,Q_4)$. Triangle areas of $\Delta(Q_2HQ_3), \Delta(Q_3HQ_4)$ and $\Delta(Q_2HQ_4)$ are proportional to the areas of the three coloured regions inside $f(P_2)$, green, yellow, and pink, respectively. The lengths of segments $HQ_1$ and $HQ_5$ are proportional to the areas $\sigma(f(P_2) \cap f(\{Q_1\}))$ and 
        $\sigma(f(P_2) \cap f(\{Q_5\}))$, respectively.
	}
	\label{fig:F6-pentagon}
\end{figure}

\section{Coordinates based on Voronoi diagrams}\label{sec:coordinatesVoronoiDiagrams}

In this section we present our generalization of Sibson's formula that expresses a point $Q_\ell \in S$ as a convex combination of points from $S$ using its neighbours of the Voronoi diagram of any given order.

\begin{theorem}\label{thm:SibsonR_k}
    Let $1\leq k\leq n-2$ and let $R_k(\ell)$ be a bounded region. Then,
    \begin{equation*}
        Q_\ell=
        \sum_{f(P_k)\in R_k(\ell) } \ \ \sum_{\substack{f(P_{k+1}) \in V_{k+1}(S)\\ Q_j \in P_{k+1}\setminus P_k}}\frac{\sigma (f(P_{k+1})\cap f(P_k))}{\sigma (R_k(\ell))}Q_j
    \end{equation*}
    
\end{theorem}
\begin{proof}
    Since $R_1(\ell)$ is the cell of $Q_\ell$ in $V_1(S)$, for $k=1$ the statement is equivalent to Theorem~\ref{thm:sibson} the formula of Sibson \cite{sibson1980vector}, i.e., the statement holds for $k=1$.
    
    Now, by induction, suppose the hypothesis is true for $R_{k-1}(\ell)$.
    By summing the equation given in Theorem~\ref{thm:Sibsonkface} for cells in $R_k(\ell)= \cup_{\ell\in P_k}f(P_k)$, we have

    \begin{equation}\label{eq:4}
        \begin{split}
            \sum_{f(P_k)\in R_k(\ell) } {\ \ \sum_{\substack{f(P_{k+1}) \in V_{k+1}(S)\\ Q_j \in P_{k+1}\setminus P_k}}\sigma (f(P_{k+1})\cap f(P_k))Q_j}& =\\
            =\sum_{f(P_k)\in R_k(\ell) } {\ \ \sum_{\substack{f(P_{k-1}) \in V_{k-1}(S)\\ Q_i \in P_k\setminus P_{k-1}}}\sigma (f(P_{k-1})\cap f(P_k))Q_i}& =\\
            =\sum_{f(P_k)\in R_k(\ell) } {\ \ \sum_{\substack{f(P_{k-1}) \in V_{k-1}(S)\\ Q_i \in P_k\setminus P_{k-1}\\Q_i\neq Q_\ell}}\sigma (f(P_{k-1})\cap f(P_k))Q_i}&\ +\\
            +\sum_{f(P_k)\in R_k(\ell) } {\ \ \sum_{\substack{f(P_{k-1}) \in V_{k-1}(S)\\ Q_\ell \in P_k\setminus P_{k-1}}}\sigma (f(P_{k-1})\cap f(P_k))Q_\ell}&\\
        \end{split}
    \end{equation}

    By properties of $R_k(\ell)$, see \cite{CdH21}, $R_k(\ell)$ without the cells that have $Q_\ell$ as its $k$-th nearest neighbour is the region of $Q_\ell$ in the previous order diagram, $R_{k-1}(\ell)$. See Figure~\ref{fig:Rk}. Then, we have that 
    
    \begin{equation*}
        \begin{split}
            &\sum_{f(P_k)\in R_k(\ell) } {\ \ \sum_{\substack{f(P_{k-1}) \in V_{k-1}(S)\\ Q_i \in P_k\setminus P_{k-1}\\Q_i\neq Q_\ell}}\sigma (f(P_{k-1})\cap f(P_k))Q_i}=\\
            =&\sum_{f(P_{k-1})\in R_{k-1}(\ell) } {\ \ \sum_{\substack{f(P_{k}) \in V_{k}(S)\\ Q_{i'} \in P_k\setminus P_{k-1}}}\sigma (f(P_k)\cap f(P_{k-1}))Q_{i'}}\\
        \end{split}
    \end{equation*}

    That is, the sum of the Lebesgue measure of the cells of $R_{k-1}(i)$ multiplied by the corresponding $k$-nearest neighbours coincides with the sum of the Lebesgue measure of the cells of $R_k(i)$, whose $k$-nearest neighbour is not $Q_\ell$, multiplied by the corresponding $k$-nearest neighbours. See Figures \ref{fig:Rk}, \ref{fig:R2} and \ref{fig:R3}.

    Since we assume that the statement is true for $R_{k-1}(\ell)$, then
    \begin{equation*}
        \sum_{f(P_{k-1})\in R_{k-1}(\ell) } {\ \ \sum_{\substack{f(P_{k}) \in V_{k}(S)\\ Q_{i'} \in P_k\setminus P_{k-1}}}\sigma (f(P_k)\cap f(P_{k-1}))Q_{i'}} =\sigma (R_{k-1}(\ell))Q_\ell
    \end{equation*}
    
    Now, replacing in Equation~(\ref{eq:4}):

    \begin{equation*}
        \begin{split}
            &\sum_{f(P_k)\in R_k(\ell) } {\ \ \sum_{\substack{f(P_{k+1}) \in V_{k+1}(S)\\ Q_j \in P_{k+1}\setminus P_k}}{\sigma (f(P_{k+1})\cap f(P_k))}Q_j}=\\
            &=\sum_{f(P_{k-1})\in R_{k-1}(\ell) } {\ \ \sum_{\substack{f(P_{k}) \in V_{k}(S)\\ Q_{i'} \in P_k\setminus P_{k-1}}}\sigma (f(P_k)\cap f(P_{k-1}))Q_{i'}}\ +\\
            &+\sum_{f(P_k)\in R_k(\ell) } {\ \ \sum_{\substack{f(P_{k-1}) \in V_{k-1}(S)\\ Q_\ell \in P_k\setminus P_{k-1}}}\sigma (f(P_{k-1})\cap f(P_k))Q_\ell}=\\
            &=\sigma (R_{k-1}(\ell))Q_\ell+\sum_{f(P_k)\in R_k(\ell) } {\ \ \sum_{\substack{f(P_{k-1}) \in V_{k-1}(S)\\ Q_\ell \in P_k\setminus P_{k-1}}}\sigma (f(P_{k-1})\cap f(P_k))Q_\ell}=\\
            &=\sigma (R_k(\ell))Q_\ell
        \end{split}
    \end{equation*}
    \begin{figure}[!ht]
    	\centering
    		\includegraphics[scale=0.5,page=5]{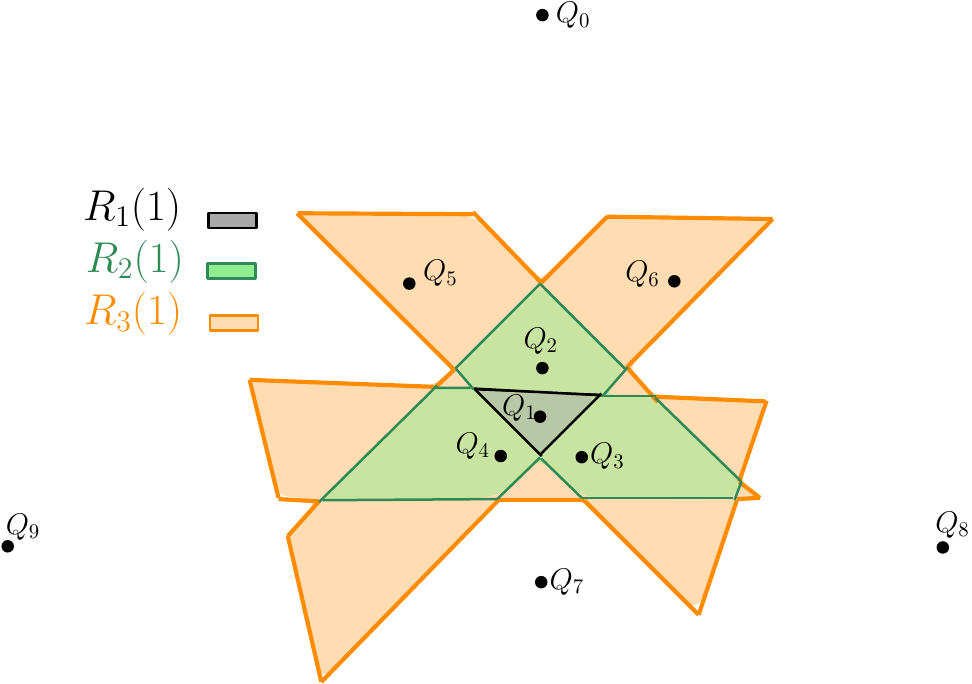}
    	\caption{For a set $S=\{Q_0,\cdots,Q_9\}$ of $10$ points. $R_1(1)\subset R_2(1)\subset R_3(1)$. Points in $R_3(1) \setminus R_2(1)$ have $Q_1$ as their third nearest neighbour. Analogously, points in $R_2(1) \setminus R_1(1)$ have $Q_1$ as their second nearest neighbour.
    	}
    	\label{fig:Rk}
    \end{figure}
    \begin{figure}[!ht]
    	\centering
        
    	\subfloat[]{
    		\includegraphics[scale=0.75,page=13]{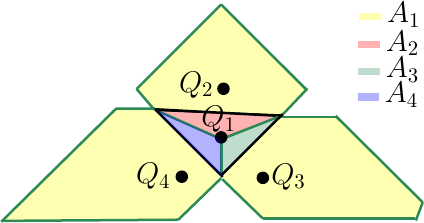}
    	}~
    	\subfloat[]{
    		\includegraphics[scale=0.75,page=10]{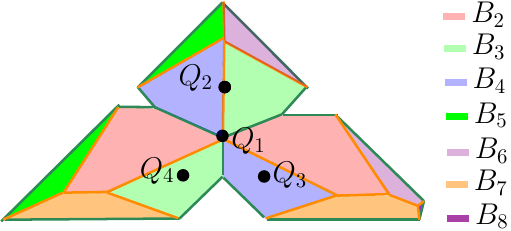}
    	}
    	\caption{
            For the same set $S$ of Figure~\ref{fig:Rk}:
    		(a) Regions $A_i$ of $R_2(1)$, are the union of cells of $V_3(S)$ whose points have $Q_i$ as the second nearest neighbour of $S$. Cells with the same second nearest neighbour get the same colour.
    		(b) Regions $B_i$ of $R_2(1)$, are the union of cells of $V_3(S)$ whose points have $Q_i$ as the third nearest neighbour of $S$. Cells with the same third nearest neighbour get the same colour.
            $\sum_{i=1\ldots 4} {\sigma(A_i)Q_i} = \sum_{i=2\ldots 8}{\sigma(B_i)Q_i}$.
    	}
    	\label{fig:R2}
 
    \end{figure}
    \begin{figure}[!ht]
    	\centering
    	\subfloat[]{
    		\includegraphics[scale=0.6,page=14]{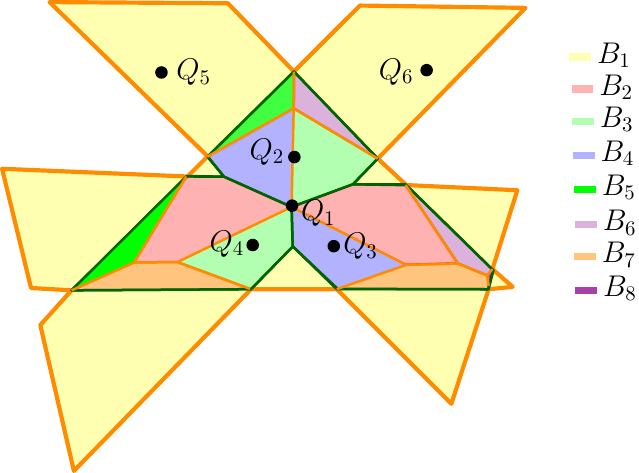}
    	}~
    	\subfloat[]{
    		\includegraphics[scale=0.6,page=11]{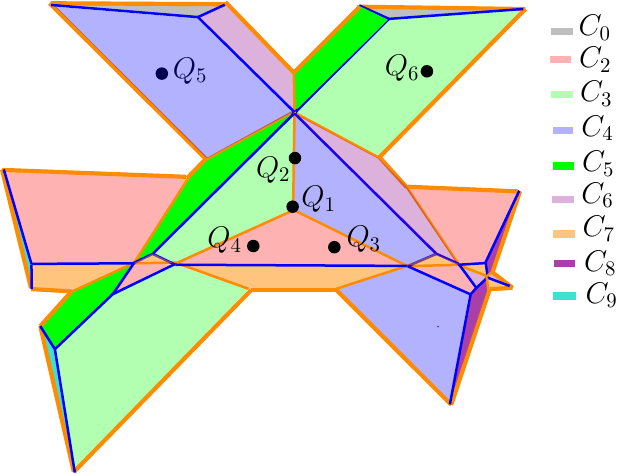}
    	}
    	\caption{
            For the same set $S$ of Figure~\ref{fig:Rk}:
            (a) Regions $B_i$ of $R_3(1)$, are the union of cells of $V_4(S)$ whose points that have $Q_i$ as the third nearest neighbour of $S$. It is shown that without $B_1$, it is the same as Figure~\ref{fig:R2} (c).
    	(b) Regions $C_i$ of $R_3(1)$, are the union of cells of $V_4(S)$ whose points have $Q_i$ as the fourth nearest neighbour of $S$.
            $\sum_{i=1\ldots 8} {\sigma(B_i)Q_i} = \sum_{i=0,2,3\ldots 9}{\sigma(C_i)Q_i}$.
    	}
    	\label{fig:R3}
    \end{figure}
\end{proof}

\newpage

\section{Towards higher order natural neighbour interpolation}\label{sec:interpolatión}

Sibson's theorem (Theorem~\ref{thm:sibson}) gave rise to the natural neighbour interpolation method. Given a set of points $S$ and known function values $G(Q_j)$ for $Q_j~\in~S\setminus\{Q_\ell\}$, the function value $G(Q_\ell)$ of a point $Q_\ell$ is interpolated by 
$G(Q_\ell)~=~\sum_j c_j G(Q_j)$,
where the sum is over the natural neighbours $Q_j$ of $Q_\ell$ in $V_1(S)$.
The local coordinates $c_j$ are given by Theorem~\ref{thm:sibson}. Note that they satisfy $\sum_j c_j =1$ and $c_j \geq 0$ for all $j$.
Then, Sibson's natural neighour interpolation is given by   
    \begin{equation}\label{eqn:interpol}
        G(Q_\ell) = \sum_{j\neq \ell}\frac{\sigma(f(\{Q_\ell,Q_j\})\cap f(\{Q_\ell\}))}{\sigma(f(\{Q_\ell\}))}G(Q_j).
    \end{equation}

The generalization of Sibson's formula given in Theorem~\ref{thm:SibsonR_k} suggests to approximate the function value $G(Q_\ell)$ by using the natural neighbours of higher order Voronoi diagrams. By using the region $R_k(\ell)$ for $k>1$, we can estimate the function value of a point $Q_\ell$ as 

\begin{equation}\label{eqn:interpolR}
    G(Q_\ell)=
    \sum_{f(P_k)\in R_k(\ell) } \ \ \sum_{\substack{f(P_{k+1}) \in V_{k+1}(S)\\ Q_j \in P_{k+1}\setminus P_k}}\frac{\sigma (f(P_{k+1})\cap f(P_k))}{\sigma (R_k(\ell))}G(Q_j)
\end{equation}

Note that  $R_1(\ell) = f(\{Q_\ell\})$ in $V_1(S)$, and for $k=1$ Equations~(\ref{eqn:interpol}) and~(\ref{eqn:interpolR}) coincide.\vspace{0.2cm}

A better estimation can be obtained by using Theorem~\ref{thm:SibsonR_k} in a combination of different values of $k$.

We explore this for the 1-dimensional case. First, we state a structural lemma.

\begin{lemma}\label{lem:1dim}
Let $S$ be set of $n$ different points on a line, and let $2 \leq k \leq n-2.$ Then, each bounded cell of $V_k(S)$ contains exactly one vertex of $V_{k-1}(S)$ and one vertex of $V_{k+1}(S)$.
\end{lemma}
\begin{proof}
Let $S=\{x_1,\ldots ,x_n\}$, where $x_1<x_2<\ldots <x_n$. 
All points of $S$ lie on a same line $L$. 
The bisectors between two consecutive points of $S$ intersect $L$ at the vertices of the Voronoi diagram of order one $V_1(S)$, that is, the points $(x_i+x_{i+1})/2, \, i=1,\ldots,n-1$. 
Analogously, the points $(x_i+x_{i+k})/2, \, i=1,\ldots,n-k$, are the vertices of $V_k(S)$.

A bounded cell of $V_k(S)$ is the segment delimited by two consecutive vertices of $V_k(S)$, $P=(x_i+x_{i+k})/2$ and $R=(x_{i+1}+x_{i+k+1})/2$.
The point $A=(x_{i+1}+x_{i+k})/2$ belongs to $V_{k-1}(S)$ and fulfills $P<A<R$.
The point $B=(x_i+x_{i+k+1})/2$ belongs to $V_{k+1}(S)$ and fulfills $P<B<R$.
Then, in the segment $PR$ we find vertex $A$ from $V_{k-1}(S)$ and vertex $B$ from $V_{k+1}(S)$. 
\end{proof}

Theorem~\ref{thm:SibsonR_k}, respectively Theorem~\ref{thm:SibsonR_k}, for dimension 1 reduces to the following statement.
\begin{property}\label{prop:line}
    Let $S=\{x_0,x_1, \ldots x_{2\ell}\}$ with $x_0 < x_1 < \ldots < x_{2\ell}$ be  real numbers. Then,

    \begin{equation}\label{eq:line}
        x_\ell = \frac{1}{x_{2\ell}-x_{0}}\left( \left(
        \sum_{i=0}^{\ell-1} x_i (x_{\ell+1+i}-x_{\ell+i})   \right)
        + \left(\sum_{i=\ell+1}^{2\ell} x_i ( x_{i-\ell} - x_{i-\ell-1})\right)\right).
    \end{equation}
\end{property}
\begin{proof}
    The bounded cells of $V_k(S)$ are intervals, bounded by midpoints between points of $S$.
    
    By Lemma~\ref{lem:1dim}, for $k>1$, each such cell contains exactly one vertex (that is, a midpoint between two points from $S$) from $V_{k+1}(S)$ and exactly one vertex from $V_{k-1}(S).$ For a given point $x_\ell$, 
    the cells $f(P_k)$ of $V_k(S)$, with $x_\ell \in P_k$, satisfy that the $k$ points from $P_k$ are consecutive points among the points $\{x_{\ell-k}, x_{\ell-k+1}, \ldots, x_{\ell+k}\}$. 
    A term $\sigma(f(P_k) \cap f(P_{k+1}))$ of Theorem~\ref{thm:SibsonR_k} corresponds to an interval with endpoints a vertex from $V_k(S)$ and another vertex from $V_{k+1}(S)$. Since vertices from $V_{k+1}(S)$ and $V_k(S)$ appear in alternating order when walking along the real line, the interval corresponding to a point  
     $x_i$ with $x_i < x_{\ell}$, has endpoints ($x_{i+1-\ell} + x_i)/2$ and  ($x_{i-\ell} + x_i)/2$. And for a  point $x_i$ with $x_i > x_{\ell}$, the corresponding term $\sigma(f(P_k) \cap f(P_{k+1}))$ of Theorem~\ref{thm:SibsonR_k} is given by the interval with endpoints ($x_i + x_{i-\ell})/2$ and  ($x_i +x_{i-\ell-1} )/2$. 
     The term $\sigma(R_k(\ell))$ in Theorem~\ref{thm:SibsonR_k} is the length of the interval with endpoints $(x_0+x_\ell)/2$ and $(x_\ell + x_{2\ell})/2$. The statement of Property~\ref{prop:line} follows.  
 \end{proof}

\begin{remark}
    Property~\ref{prop:line} has actually a more general statement. The assumption $x_0 < x_1 < \ldots < x_{2\ell}$ is not needed.
    This follows easily: When expending the terms of the two sums on the right side of Equation~(\ref{eq:line}), all terms cancel out except $x_\ell x_{2\ell}$ and $-x_\ell x_0$. The given proof using $V_k(S)$ shows that Equation~(\ref{eq:line}) indeed is Theorem~\ref{thm:SibsonR_k} in dimension 1.
\end{remark}

We denote points $Q_i$ of $S$ as $x_i$ and their function values $G(Q_i)$ as $y_i$. When $k=1$ we have Sibson's classical nearest neighbour interpolation, which for dimension $d=1$ is piecewise linear interpolation. Let $x_0, x_1, \ldots, x_5$ be six points on the real line in that order. And let $x_2<x<x_3$ be a query point whose function value $G(x)$ we want to interpolate. To avoid degenerate cases where bisectors between points coincide, we also assume that all midpoints $(x_i+x_j)/2$ with $x_i, x_j \in \{S \cup \{x\}\}$ are different.
Sibson's classical formula, Equation~(\ref{eqn:interpol}), uses the two neighbours $x_2$ and $x_3$ of $x$, and gives the interpolation
\begin{equation}\label{eqn:sibson1d}
    G_1(x) = \frac{1}{x_3-x_2}\left( y_2(x_3-x) + y_3 (x-x_2)  \right),
\end{equation}
i.e. point $(x,G_1(x))$ lies on the line segment connecting points $(x_2,y_2)$ and $(x_3,y_3).$ This can also be deduced from  Property~\ref{prop:line}.

For $k=2$, from Equation~(\ref{eqn:interpolR}) we obtain
\begin{equation}\label{eqn:gensibson2}
    G'_2(x)=\frac{1}{x_4-x_1}\left( y_1(x_3 -x) +y_2(x_4-x_3)+y_3(x_2-x_1)+y_4(x-x_2)    \right)
\end{equation}
and for $k=3$,
\begin{equation}\label{eqn:gensibson3}
    \begin{split}
        G'_3(x)=&\frac{1}{x_5-x_0}\left( y_0(x_3 -x) +y_1(x_4-x_3)+y_2(x_5-x_4) \right.\\
         & \left.+y_3(x_1-x_0)+y_4(x_2-x_1)+y_5(x-x_2).\right)\\
    \end{split}
\end{equation}

Note that $x$ only appears in the first and the last summand in Equations~(\ref{eqn:gensibson2}) and~(\ref{eqn:gensibson3}). We therefore add
$\frac{x_3-x_2}{x_4-x_1}G_1(x)$
 to Equation~(\ref{eqn:gensibson2})
and obtain
\begin{equation*}
    G'_2(x)+\frac{x_3-x_2}{x_4-x_1}G_1(x)=\frac{1}{x_4-x_1}(y_1(x_3-x) + y_2(x_4-x)+y_3(x-x_1) + y_4(x-x_2) ).
\end{equation*}
 By estimating $G(x) = G_1(x)$ and also $G(x)=G'_2(x)$ we obtain $G(x)\frac{x_4-x_1+x_3-x_2}{x_4-x_1}$ on the left hand side of this equation, and then the estimate, for $k=2$,
\begin{equation}\label{eqn:interpol2}
    G_2(x)=\frac{1}{x_4-x_1 + x_3-x_2}(y_1(x_3-x) + y_2(x_4-x)+y_3(x-x_1) + y_4(x-x_2) ).
\end{equation}
In the same way, by adding the equations $G'_3(x)+\frac{x_4-x_1}{x_5-x_0}G'_2(x)+ \frac{x_3-x_2}{x_5-x_0}G_1(x),$ we get the estimate for $k=3,$
\begin{equation}\label{eqn:interpol3}
    \begin{split}
        G_3(x)=&\frac{1}{x_5-x_0+x_4-x_1+x_3-x_2  } (  y_0(x_3-x) +y_1(x_4-x)+ y_2(x_5-x)\\
        &+y_3(x-x_0)+ y_4(x-x_1)+y_5(x-x_2) ).
    \end{split}
\end{equation}

Figure~\ref{fig:sib1d} shows an example of the interpolation formulas given in Equations~(\ref{eqn:sibson1d}), (\ref{eqn:interpol2}), and~(\ref{eqn:interpol3}).

\begin{figure}[!ht]
	\centering	
	\includegraphics[scale=0.3, trim= 0cm 4cm 0cm 0cm, clip]{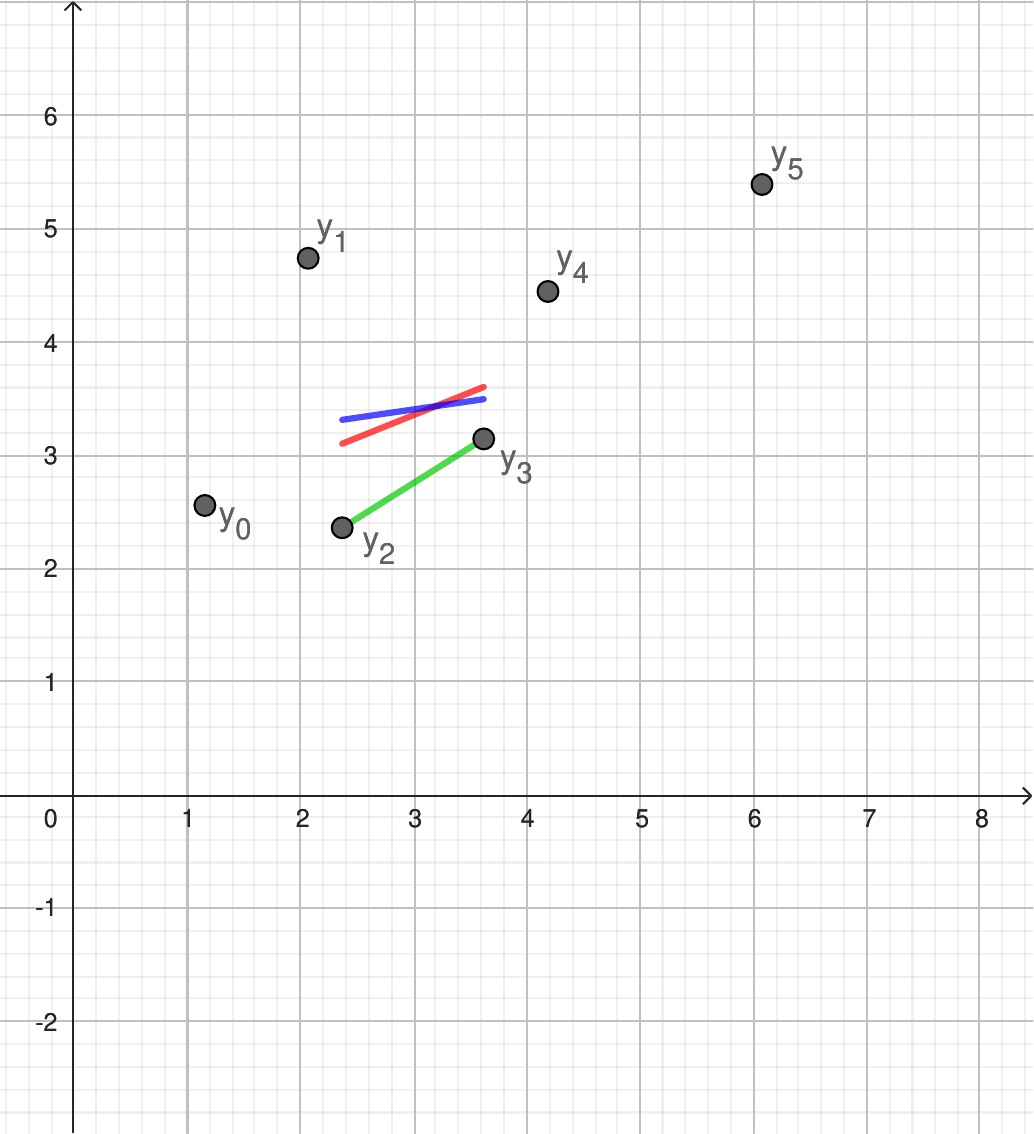}
	\caption{The generalized Sibson interpolation in $\mathbb{R}^1$. In green: Sibson's original interpolation, Equation~(\ref{eqn:sibson1d}), used only $R_1(x)$.
    The blue segment shows the interpolation using $R_1(x)$ and $R_2(x)$, given by Equation~(\ref{eqn:interpol2}). Four points are used. The red segment shows the interpolation using $R_1(x),$ $R_2(x),$ and $R_3(x)$, given by Equation~(\ref{eqn:interpol3}). 
    Six points are used.  
    }
 \label{fig:sib1d}	
\end{figure}

We conclude with some comments on the proposed interpolation formulas. First, they appear in a natural way from the generalization of Sibson's formula. This already makes it worth to study  such generalized interpolation formulas. 
In Equations~(\ref{eqn:interpol2}) and~(\ref{eqn:interpol3}), the coefficients $c_j$ in $G_i(x)=\sum_j c_j y_j$, $i=2,3$, satisfy $\sum_j c_j=1$ and $c_j \geq 0$ for every $c_j.$ 
We also mention that it can not be guaranteed that  $G_i(x)$ coincides with $G_i(x_2)$ or with $G_i(x_3)$, when $x$ coincides with one of the endpoints of the interval, $x_2$ or $x_3$, respectively. Though, we observe that in this case, the point farthest away from $x$ on one side, drops from being used in the interpolation formula. This also holds for the classical case $k=1.$

Finally, we expect that the generalized interpolation formulas can have applications. For instance, when the used values for the interpolation are obtained by measurements and measurement inaccuracy can not be ruled out. Then reliability might be improved by using  nearest neighbours from $V_k(S)$ or by using $R_k(x)$, instead of only $V_1(S)$.

\bibliographystyle{plainurl}
\bibliography{bibliografia}

\begin{thebibliography}{10}

\bibitem{aurenhammer1988linear}
Franz Aurenhammer.
\newblock Linear combinations from power domains.
\newblock {\em Geometriae dedicata}, 28(1):45--52, 1988.

\bibitem{aurenhammer1991simple}
Franz Aurenhammer and Otfried Schwarzkopf.
\newblock A simple on-line randomized incremental algorithm for computing higher order {V}oronoi diagrams.
\newblock In {\em Proceedings of the seventh annual symposium on Computational geometry}, pages 142--151, 1991.

\bibitem{CdH21}
Merc{\`e} Claverol, Andrea de las~Heras Parrilla, Clemens Huemer, and Alejandra Mart{\'\i}nez-Moraian.
\newblock The edge labeling of higher order {V}oronoi diagrams.
\newblock 2021.
\newblock \url{https://arxiv.org/abs/2109.13002}.

\bibitem{edelsbrunner2018multiple}
Herbert Edelsbrunner and Mabel Iglesias-Ham.
\newblock Multiple covers with balls i: Inclusion--exclusion.
\newblock {\em Computational Geometry}, 68:119--133, 2018.

\bibitem{farin1990surfaces}
Gerald Farin.
\newblock Surfaces over {D}irichlet tessellations.
\newblock {\em Computer aided geometric design}, 7(1-4):281--292, 1990.

\bibitem{jones1984geometric}
Gareth~A Jones.
\newblock Geometric and asymptotic properties of {B}rillouin zones in lattices.
\newblock {\em Bulletin of the London Mathematical Society}, 16(3):241--263, 1984.

\bibitem{lee1982Voronoi}
Der-Tsai Lee.
\newblock On k-nearest neighbor {V}oronoi diagrams in the plane.
\newblock {\em IEEE Transactions on Computers}, C-31(6):478--487, 1982.

\bibitem{mammana2008affine}
Maria~Flavia Mammana and Biagio Micale.
\newblock Quadrilaterals of triangle centres.
\newblock {\em The Mathematical Gazette}, 92:466--475, 2008.

\bibitem{okabe1994nearest}
Atsuyuki Okabe, Barry Boots, and Kokichi Sugihara.
\newblock Nearest neighbourhood operations with generalized {V}oronoi diagrams: a review.
\newblock {\em International Journal of Geographical Information Systems}, 8(1):43--71, 1994.

\bibitem{okabe2009spatial}
Atsuyuki Okabe, Barry Boots, Kokichi Sugihara, and Sung~Nok Chiu.
\newblock Spatial tessellations: concepts and applications of {V}oronoi diagrams.
\newblock 2009.

\bibitem{piper1993properties}
Bruce~R. Piper.
\newblock Properties of local coordinates based on {D}irichlet tessellations.
\newblock In {\em Geometric modelling}, pages 227--239. Springer, 1993.

\bibitem{radko20128bisector}
Olga Radko and Emmanuel Tsukerman.
\newblock The perpendicular bisector construction, the isotopic point, and the {S}imson line of a quadrilateral.
\newblock {\em Forum Geometricorum}, 12:161--189, 2012.

\bibitem{sibson1980vector}
Robin Sibson.
\newblock A vector identity for the {D}irichlet tessellation.
\newblock In {\em Mathematical Proceedings of the Cambridge Philosophical Society}, volume~87, pages 151--155. Cambridge University Press, 1980.

\bibitem{sibson1981brief}
Robin Sibson.
\newblock A brief description of natural neighbour interpolation.
\newblock {\em Interpreting multivariate data}, pages 21--36, 1981.

\bibitem{sugihara1999surface}
Kokichi Sugihara.
\newblock Surface interpolation based on new local coordinates.
\newblock {\em Computer-Aided Design}, 31(1):51--58, 1999.

\bibitem{toth1976multiple}
G{\'a}bor~Fejes T{\'o}th.
\newblock Multiple packing and covering of the plane with circles.
\newblock {\em Acta Math. Acad. Sci. Hungar}, 27(1-2):135--140, 1976.

\bibitem{veerman2000brillouin}
JJP Veerman, Mauricio~M Peixoto, Andr{\'e}~C Rocha, and Scott Sutherland.
\newblock On {B}rillouin zones.
\newblock {\em Communications in Mathematical Physics}, 212:725--744, 2000.

\end{thebibliography}

\end{document}